\documentclass[journal, ,onecolumn,12pt]{IEEEtran}
\ifCLASSINFOpdf
\else
   \usepackage[dvips]{graphicx}
\fi
\usepackage{url}

\usepackage{amsmath}
\usepackage{amsthm}
\usepackage{amssymb}
\usepackage{amsfonts}
\usepackage{algorithm}
\usepackage{algpseudocode}
\usepackage{xcolor}
\usepackage{comment}
\usepackage{graphicx}
\usepackage{bbm}

\newtheorem{theorem}{Theorem}

\begin{document}

\title{Reliable Sub-Nyquist Spectrum Sensing via Conformal Risk Control}

\author{Hyojin Lee, \IEEEmembership{Student Member, IEEE}, Sangwoo Park, \IEEEmembership{Member, IEEE}, Osvaldo Simeone, \IEEEmembership{Fellow, IEEE}, Yonina C. Eldar, \IEEEmembership{Fellow, IEEE}, and Joonhyuk Kang, \IEEEmembership{Member, IEEE}
\thanks{This work was supported in part by the Ministry of Science and ICT (MSIT), South Korea, through the Information Technology Research Center (ITRC) Support Program supervised by the Institute of Information and Communications Technology Planning and Evaluation (IITP) under Grant IITP-2024-2020-0-01787 and Grant IITP-2024-RS-2023-00259991. The work of O. Simeone is supported by the European Union’s Horizon Europe project CENTRIC (101096379), by an Open Fellowship of the EPSRC (EP/W024101/1), and by the EPSRC project (EP/X011852/1). The work of Y. C. Eldar is supported by the European Research Council (ERC) under the European Union’s Horizon 2020 research and innovation program (grant No. 101000967) and by the Israel Science Foundation (grant No. 536/22).}
\thanks{H. Lee and J. Kang are with the School of Electrical Engineering, Korea Advanced Institute of Science and Technology (KAIST), Daejeon 34141, South Korea (e-mail: hyojin.lee@kaist.ac.kr; jhkang@ee.kaist.ac.kr).}
\thanks{S. Park and O. Simeone are with King’s Communication Learning \& Information Processing (KCLIP) lab, Centre for Intelligent Information Processing Systems (CIIPS), at the Department of Engineering, King’s College London, U.K. (email: {sangwoo.park; osvaldo.simeone}@kcl.ac.uk).}
\thanks{Y. C. Eldar is with the Faculty of Math and CS, Weizmann Institute of Science, Rehovot, Israel (e-mail: yonina.eldar@weizmann.ac.il).}}

\markboth{}
{}
\maketitle

\begin{abstract}
Detecting occupied subbands is a key task for wireless applications such as unlicensed spectrum access. Recently, detection methods were proposed that extract per-subband features from sub-Nyquist baseband samples and then apply thresholding mechanisms based on held-out data. Such existing solutions can only provide guarantees in terms of false negative rate (FNR) in the asymptotic regime of large held-out data sets. In contrast, this work proposes a threshold mechanism-based conformal risk control (CRC), a method recently introduced in statistics. The proposed CRC-based thresholding technique formally meets user-specified FNR constraints, irrespective of the size of the held-out data set.  By applying the proposed CRC-based framework to both reconstruction-based and classification-based sub-Nyquist spectrum sensing techniques, it is verified via experimental results that CRC not only provides theoretical guarantees on the FNR but also offers competitive true negative rate (TNR) performance.
\end{abstract}


\IEEEpeerreviewmaketitle

\section{Introduction}
\IEEEPARstart{S}{pectrum} sharing aims at increasing the efficiency of bands already allocated to wireless communications by allowing for opportunistic reuse \cite{sharing, ss2, ss3, ss4}. The most challenging part of spectrum sharing is the detection of idle bands, which is often referred to as \emph{spectrum sensing} \cite{ss}. It is particularly critical for spectrum sensing to be carried out \emph{reliably} in order to avoid interference towards active communication links. This work presents a solution to this problem that ensures assumption-free theoretical reliability guarantees.

In order for spectrum sensing to operate on large bandwidths, conventional sampling at the Nyquist rate may be excessively costly in terms of computational energy due to the need to implement fast analog-to-digital converters (ADCs). \emph{Sub-Nyquist} sampling is a well-studied solution to this problem \cite{landau,venk,blind, yonina, reviewyonina}. The classical work \cite{landau} showed that the sampling rate required for perfect reconstruction of a signal can be smaller than the Nyquist sampling rate in the presence of unused frequency subbands. Furthermore, it has been more recently demonstrated that the sampling rate can be further reduced when the goal is to reconstruct the power spectral density of the signal \cite{yonina}.

\begin{figure}
    \centerline{\includegraphics[width=13cm]{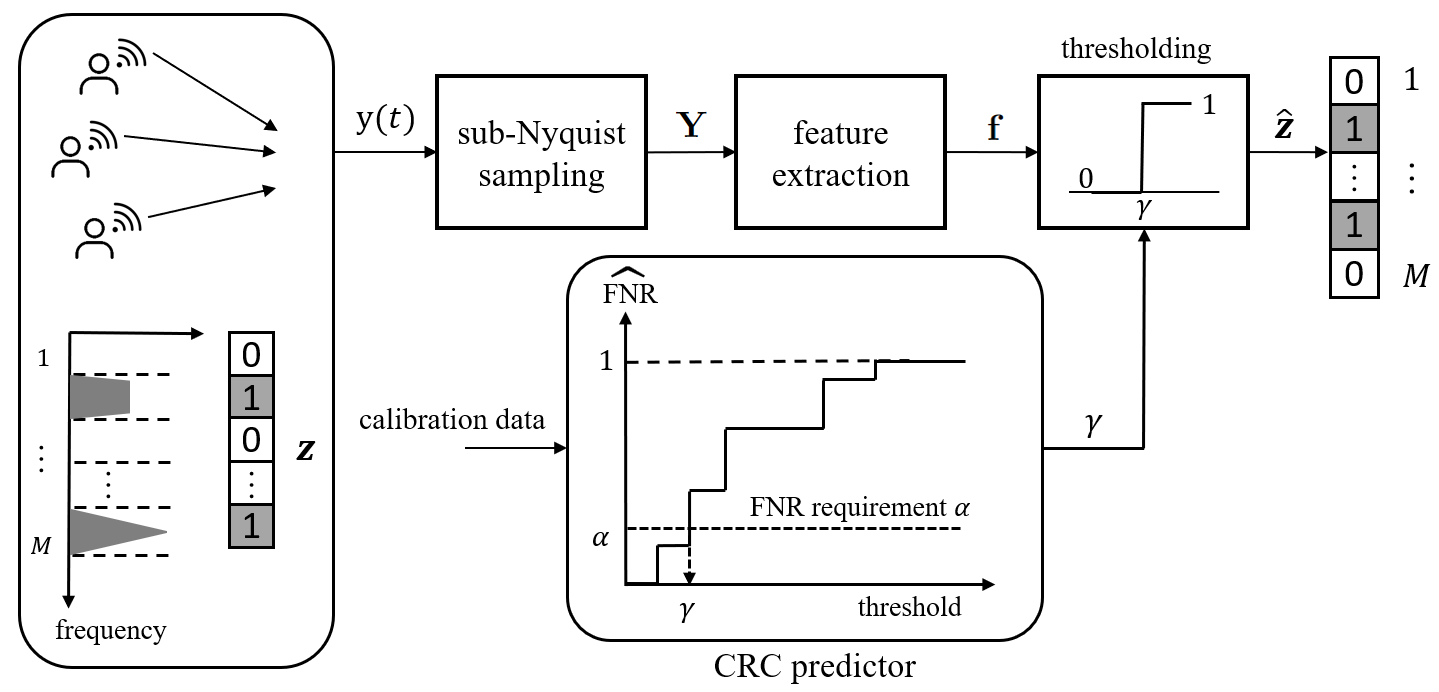}}
    \caption{In sub-Nyquist spectrum sensing, a vector of features $\mathbf{f}$ is extracted from a sub-Nyquist sample matrix $\mathbf{Y}$, from which a predicted set of occupied subbands, encoded by binary vector $\mathbf{\hat{z}}$, is obtained via thresholding. Using calibration data, the proposed approach, based on conformal risk control (CRC), selects thresholds that guarantee any desired target level $\alpha$ of the false negative rate (FNR) on average, irrespective of the choice of the feature extractor.}
    \label{diag1}
\end{figure}

State-of-the-art studies on sub-Nyquist spectrum sensing have focused on finding ways to extract more informative \emph{per-subband features} from sub-Nyquist samples. Most notably reference \cite{yonina} proposed methods to reconstruct the power spectral density, while \cite{gao} proposed to use a feature vector extracted from a neural network trained to carry out parallel binary classification tasks across the subbands (see also \cite{ensen} for an application to modulation classification). By varying the thresholds, one can trace a trade-off between the \emph{false negative rate} (FNR) -- declaring a subband free when it is occupied -- and the \emph{true negative rate} (TNR) -- deciding that a subband is free when it is idle. In typical applications, such as unlicensed spectrum access, one wishes to enforce an upper bound on the FNR. In prior works \cite{yonina}, \cite{vincent, ssedvin, ssedman, quantile}, this is done by leveraging held-out data via non-parameteric or parameteric quantile estimators. Such existing solutions can only provide FNR guarantees that apply to the \emph{asymptotic} regime of large held-out data sets.

This work proposes a threshold mechanism based \emph{conformal risk control} (CRC) \cite{crc}, a method recently introduced in statistics as a generalization of conformal prediction \cite{vovk2005algorithmic, angelopoulos2023conformal}. The proposed CRC-based thresholding technique formally meets user-specified FNR constraints, irrespective of the size of the held-out data set. By applying the proposed CRC-based framework to both reconstruction-based \cite{yonina} and classification-based \cite{gao} sub-Nyquist spectrum sensing techniques, it is verified via experimental results that CRC not only provides theoretical guarantees on the FNR but also offers competitive TNR performance.

The rest of the paper is organized as follows. Sec. II introduces the system model. In Sec. III, we establish the goal of this paper and review existing thresholding schemes. We propose the CRC-based thresholding method in Sec. IV. Numerical results are presented in Sec. V, and we conclude in Sec. VI.

\section{System Model}
\subsection{Signal Model and Sampling}
As illustrated in Fig. \ref{diag1}, we consider a spectral band divided into $M$ subbands of bandwidth $B$ [Hz], with each subband being either idle or occupied. Let $z_m\in\{0,1\}$ denote the occupancy status of the $m$-th subband. The distribution of the vector $\mathbf{z}=[z_1,\dots,z_M]^T$ is defined by the joint probability mass function $p(\mathbf{z})$.

Denote as $s_m(t)$ the baseband signal that may occupy the $m$-th subband, accounting also for the effect of the propagation channel. The real-valued multiband continuous-time baseband received signal $y(t)$ is expressed as
\begin{equation}
    y(t)=\sum_{m=1}^M z_ms_m(t)+n(t),
    \label{sig}
\end{equation}
where $n(t)$ is additive white Gaussian noise (AWGN) with power spectral density $N_0$. The multiband signal (\ref{sig}) is sampled by the receiver at a sub-Nyquist rate, i.e., at a rate smaller than the \emph{Nyquist rate} $1/T=2MB$.

\begin{figure}
    \centerline{\includegraphics[width=9cm]{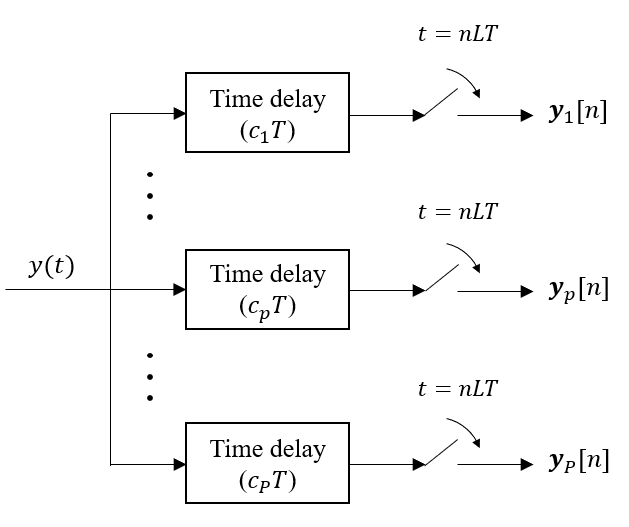}}
    \caption{Illustration of multicoset sampling (MCS), which processes the input signal $y(t)$ on $P>1$ sub-Nyquist sampling branches, each applying a distinct delay. Every sampling branch samples the delayed signal at the same sub-Nyquist sampling rate of $1/LT$, with $L>P$, given the Nyquist sampling rate $1/T$.}
    \label{mcs}
\end{figure}

\textit{Multicoset sampling} (MCS) is a sub-Nyquist sampling scheme for sparse multiband signal \cite{yonina}, \cite{blind}, \cite{gao}. As seen in Fig. \ref{mcs}, MCS employs $P$ parallel sub-Nyquist sampling branches, known as \textit{cosets}, each characterized by a distinct time delay. Each coset $p=1,\dots,P$ delays the signal (\ref{sig}) by $c_pT$, where $c_p$ is an integer, and then it samples the delayed signal at a rate $1/LT=2MB/L$ with $L>P$. The integers $\{c_p\}_{p=1}^P$ satisfy the inequalities $0\leq c_1<c_2<\dots<c_P<L$. 

Overall, denoting as $N_s$ the number of samples collected by each coset during a \emph{sensing period}, the sampled sequence produced by the $p$-th branch in a sensing period is given as
\begin{equation}
    \mathbf{y}_p[n]=y(nLT+c_pT), \qquad n=1,\dots,N_s.
    \label{sampling}
\end{equation}
Finally, the sub-Nyquist samples are collected into a $P\times N_s$ matrix $\mathbf{Y}=[\mathbf{y}_1,\dots,\mathbf{y}_P]^T$, where $\mathbf{y}_p=[\mathbf{y}_p[1],\dots,\mathbf{y}_p[N]]^T$ represents the vector of samples (\ref{sampling}) for the $p$-th branch. The total number of received samples is $N=PN_s$. Based on the sampled signal, the receiver aims to estimate the spectral occupancy vector $\mathbf{z}$ while providing guarantees on the FNR.

We note that alternative sub-Nyquist sensing methods, such as the modulated wideband converter \cite{mishali2010theory, dounaevsky2011xampling} could be also considered. Our focus in this work is on the thresholding process, which is introduced next.

\subsection{Estimating the Spectrum Occupancy Vector}
In sub-Nyquist spectrum sensing, the receiver collects the sub-Nyquist matrix $\mathbf{Y}$ and extracts a feature vector $\mathbf{f}=[f_1,\dots,f_M]^T\in\mathbb{R}^M$ to facilitate estimation of the spectrum occupancy vector $\mathbf{z}=[z_1,\dots,z_M]^T$. In this paper, we exploit two main approaches to obtain the feature vector $\mathbf{f}$.

\subsubsection{Power spectral density (PSD) \cite{yonina}}
Using the discrete Fourier transform (DFT) of the sample matrix $\mathbf{Y}$ and employing a recovery algorithm based on compressive sensing, a feature vector $\mathbf{f}$ can be extracted in which the $m$-th element, $f_m \geq0$, corresponds to the estimated received power within the $m$-th subband \cite{yonina}.

\subsubsection{Logit vector (LV) \cite{gao}}
Using a feed-forward neural network (FFNN) that takes as input a vectorized version of $\mathbf{Y}$, a feature vector $\mathbf{f}$ can be obtained as the output of the last layer of the FFNN. As detailed in \cite{gao}, the FFNN is trained to minimize a cross-entropy loss, and each output $f_m\in [0,1]$ represents the probability assigned by the FFNN to subband $m$ being occupied.

Once the vector $\mathbf{f}$ is obtained using either of the approaches above, for each $m$-th subband, the corresponding feature $f_m$ is compared with a \emph{threshold} $\gamma_m$ to decide whether the subband is occupied or not. Accordingly, the \emph{estimated spectrum occupancy vector} is given by $\hat{z}=[\hat{z}_1,\dots,\hat{z}_M]^T$ with
\begin{equation}
    \hat{z}_m=\mathbbm{1}(f_m\geq \gamma_m),
    \label{dec}
\end{equation}
where $m=1,\dots,M$, and $\mathbbm{1}(\cdot)$ denotes the indicator function that returns $1$ if its argument satisfies the condition inside of the function and $0$ otherwise. The design of the thresholds $\boldsymbol{\gamma}=[\gamma_1,\dots,\gamma_M]^T$ is the focus of this work and will be discussed in the next section.

\section{Calibration}
Calibration is the process of evaluating the threshold vector $\boldsymbol{\gamma}$ in (\ref{dec}) in order to satisfy reliability requirements \cite{vincent}, \cite{ssedvin}, \cite{quantile}. In this section, we first explain the problem of calibration, and then review existing solutions.

\subsection{Design Goal}
Denote as $p(\mathbf{Y,z})$ the joint distribution of the $P\times N_s$ received sample matrix $\mathbf{Y}=[\mathbf{y}_1,\dots,\mathbf{y}_P]^T$ and of the $M\times 1$ true occupancy vector $\mathbf{z}=[z_1,\dots,z_M]^T$. Note that this distribution depends on the unknown probability $p(\mathbf{z})$ of the spectrum occupancy vector, on the distribution of the transmitted signal $s_m(t)$, which may be also unknown, and on the noise $n(t)$ in (\ref{sig}). \emph{Calibration} assumes the availability of a calibration data set $\mathcal{D}^\text{cal}=\{\mathbf{Y}^\text{cal}_i, \mathbf{z}^\text{cal}_i \}_{i=1}^{|\mathcal{D}^\text{cal}|}$, where the $i$-th calibration example is composed of the sub-Nyquist sample matrix $\mathbf{Y}^\text{cal}_i$ and the corresponding true spectrum occupancy vector $\mathbf{z}^\text{cal}_i$, which are assumed to follow the joint distribution $p(\mathbf{Y,z})$. Using the data set $\mathcal{D}^\text{cal}$, the goal is to extract a threshold vector $\boldsymbol{\gamma}$ for the decision rule (\ref{dec}) that satisfies a constraint on FNR for a new data sample $(\mathbf{Y,\mathbf{z}})\sim p(\mathbf{Y,z})$.

The \emph{false negative rate} (FNR) is defined as the fraction of occupied subbands that are mistakenly identified as free, i.e.,
\begin{equation}
    \text{FNR}=\frac{\sum_{m=1}^M (1-\hat{z}_m)z_m}{\sum_{m=1}^M z_m}.
    \label{fnr}
\end{equation}
We impose the constraint that the expected FNR must be less than or equal to a user-defined level $\alpha\in[0,1]$, i.e.,
\begin{equation}
    \mathbb{E}[\text{FNR}]\leq \alpha,
    \label{goal}
\end{equation}
where the expectation is taken over the joint distribution $p(\mathbf{Y,z})$.

Satisfying inequality (\ref{goal}) is trivial, as one can set $\hat{z}_{m}=1$ for all subbands $m=1,\dots,M$, thus not missing any of the occupied subbands. Therefore, it is important to also evaluate the \emph{true negative rate} (TNR), which quantifies the average fraction of unoccupied subbands that are correctly identified:
\begin{equation}
    \text{TNR}=\frac{\sum_{m=1}^M (1-\hat{z}_m)(1-z_m)}{\sum_{m=1}^M (1-z_m)}.
\end{equation}

\subsection{Existing Solutions}
We now review existing solutions that attempt to satisfy constraint (\ref{goal}). Given the calibration data set $\mathcal{D}^\text{cal}$, the feature vector $\mathbf{f}_i^\text{cal}=[f_{i,1}^\text{cal},\dots,f_{i,M}^\text{cal}]^T$ is extracted from matrix $\mathbf{Y}_i^\text{cal}$ for each $i$-th data sample using either the PSD or LV approach. Then, separately for each subband $m$, existing schemes select all the features $f_{i,m}^\text{cal}$ for which the $m$-th subband is occupied, i.e., for which we have $z_{i,m}=1$. This yields the sets of features $\mathcal{F}_m^\text{cal}=\{f_{i,m}^\text{cal}\}_{i:z_{i,m}^\text{cal}=1}$ for each subband $m=1,\dots,M$.

\subsubsection{Parameteric thresholding \cite{vincent}, \cite{ssedvin}}
Parameteric thresholding methods estimate the mean $\mu_m$ and the standard deviation $\sigma_m$ of the features $f_m$ under the assumption that the subband $m$ is occupied. This is done by evaluating the empirical average of the elements in $\mathcal{F}_m^\text{cal}$. Then, the threshold $\gamma_m$ is set as the $\alpha$-quantile of the estimated Gaussian distribution, i.e., as $\gamma_m=\sigma_m Q^{-1}(1-\alpha)+\mu_m$, where $Q(\cdot)$ denotes the Q-function of the standard Gaussian distribution.

\subsubsection{Non-parameteric thresholding \cite{quantile}, \cite{q}}
Without attempting to model the feature distribution, non-parameteric methods set the decision threshold $\gamma_m$ directly as the $\alpha$-quantile of the distribution of features $f_m$ under the assumption that the subband $m$ is occupied. This is done by estimating the empirical $\alpha$-quantile of the elements of set $\mathcal{F}_m^\text{cal}$. Accordingly, the $m$-th threshold $\gamma_m$ is chosen as the $\lfloor \alpha|\mathcal{F}^\text{cal}|\rfloor$-th smallest value in $\mathcal{F}^\text{cal}$, where $|\mathcal{F}^\text{cal}|$ is the size of the set $\mathcal{F}^\text{cal}$.

The rationale of both approaches is that, if the parameteric or non-parameteric estimate of the $\alpha$-quantile is exact, then the probability of the feature $f_m$ being below the threshold $\gamma_m$, yielding the estimate $\hat{z}_m=0$ in (\ref{dec}), when the subband is occupied, i.e., when $z_m=1$, is no larger than $\alpha$. This would meet the reliability constraint (\ref{goal}). In this regard, note that the parameteric model may potentially be more sample-efficient, but it is adversely affected when the unknown feature distribution is not Gaussian.

\section{CRC-Based FNR Risk Control}
Both existing parameteric and non-parameteric methods reviewed in the previous section can only guarantee the FNR requirement (\ref{goal}) under the assumption of an exact estimate of the quantiles. This assumption, in turn, can be verified only in the asymptotic regime of an infinitely large calibration data set, i.e., when $|\mathcal{D}^\text{cal}|\rightarrow \infty$. In this paper, we propose a non-parameteric thresholding method that can provably meet the FNR constraint (\ref{goal}), irrespective of the true distribution $p(\mathbf{Y,z})$ and of the size $|\mathcal{D}^\text{cal}|$ of the calibration data set.

Unlike the existing methods reviewed in the previous section, the proposed method leverages information across all subbands to determine a common threshold $\gamma_1=\gamma_2=\dots=\gamma_M=\gamma$. The main idea is to set the common threshold $\gamma$ to ensure that a suitably corrected version of the empirical FNR, estimated using the calibration set $\mathcal{D}^\text{cal}$, is no larger than the target level $\alpha$. By choosing the correction as dictated by CRC \cite{crc}, we show that condition (\ref{goal}) can be met without making any assumptions on the ground-truth distribution $p(\mathbf{Y,z})$.

To start, for each $i$-th calibration point $(\mathbf{Y}_i^\text{cal},\mathbf{z}_i^\text{cal})$, we evaluate the FNR as a function of a threshold $\gamma$ as follows. Given a threshold $\gamma$, we estimate the occupancy vector as $\hat{\mathbf{z}}_i^\text{cal}=[\hat{z}_{i,1},\dots,\hat{z}_{i,M}]^T$, where
\begin{equation}
    \hat{z}_{i,m}(\gamma)=\mathbbm{1}(f_{i,m}^\text{cal}\geq \gamma),\quad m=1,\dots,M
    \label{dec1}
\end{equation}
is the occupancy estimate for the $m$-th subband.

The empirical FNR for the $i$-th calibration data point is then calculated by following (\ref{fnr}) as
\begin{equation}
    \text{FNR}_i(\gamma)=\frac{\sum_{m=1}^M (1-\hat{z}_{i,m}^\text{cal}(\gamma))\cdot z_{i,m}^\text{cal}}{\sum_{m=1}^M z_{i,m}^\text{cal}}.
    \label{loss}
\end{equation}
By (\ref{dec1}), function $\text{FNR}_i(\gamma)$ is non-decreasing in the threshold $\gamma$. In particular, we have the inequalities $0\leq \text{FNR}_i(\gamma)\leq 1$, with $\text{FNR}_i(0)=0$. In fact, with $\gamma=0$, all subbands are estimated as occupied, and the FNR equals zero. Furthermore, as $\gamma$ increases, $\text{FNR}_i(\gamma)$ tends to 1.

We evaluate the empirical average FNR over the calibration data set $\mathcal{D}^\text{cal}$ as
\begin{equation}
    \overline{\text{FNR}}(\gamma|\mathcal{D^\text{cal}})=\frac{1}{\mathcal{|D^\text{cal}}|} \sum_{i=1}^{|\mathcal{D^\text{cal}}|} \text{FNR}_i(\gamma).
    \label{empfnr}
\end{equation}
The proposed method chooses the common threshold $\gamma$ as
\begin{equation}
    \hat{\gamma}(\mathcal{D}^\text{cal})=\textrm{sup}\ \left\{\gamma:\frac{|\mathcal{D}^\text{cal}|}{|\mathcal{D}^\text{cal}|+1}\overline{\text{FNR}}(\gamma|\mathcal{D}^\text{cal})+\frac{1}{|\mathcal{D}^\text{cal}|+1}\leq\alpha\right\}.
    \label{crc}
\end{equation}

Following CRC, the threshold $\hat{\gamma}(\mathcal{D}^\text{cal})$ is obtained by correcting the empirical FNR in (\ref{empfnr}) adding a term, $1/(|\mathcal{D}^\text{cal}|+1)$, which can be interpreted as the maximum value of the FNR attained on a fictitious additional calibration point. Note that, with the added fictitious data point, the calibration data set efficiently includes $|\mathcal{D}^\text{cal}|+1$ examples. This correction is instrumental in guaranteeing the following property.

\begin{theorem} \label{ther_1}
    Assume that the $|\mathcal{D}^\text{cal}|$ calibration samples in the calibration data set \(\mathcal{D}^\text{cal}\) and the test sample \((\mathbf{Y},\mathbf{z})\) are generated in an i.i.d. manner from distribution \(p(\mathbf{Y,z})\). Then, the inequality  
    \begin{equation}
        \mathbb{E}[{\textup{FNR}}(\hat{\gamma}(\mathcal{D}^\text{cal})|\mathcal{D}^\text{cal})]\leq \alpha
        \label{goal3}
    \end{equation}
    holds for any $\alpha\in [0,1]$, where the expectation is taken of the joint distribution over calibration and test samples.
\end{theorem}

\begin{proof}
    The proof follows from \cite[Theorem 1]{crc}.
\end{proof}

Thus, the proposed CRC-based thresholding for sub-Nyquist spectrum sensing can provably guarantee the FNR constraint (\ref{goal}) with any desired level $\alpha$. The guarantee in Theorem~\ref{ther_1} holds irrespective of the sampling rate, the feature extraction strategy, the number of calibration samples, and the characteristics of the observation noise.

\section{Numerical Results}
In this section, the proposed CRC-based thresholding method is compared with the existing thresholding methods reviewed in Sec. II for both PSD and LV features.

\subsection{Simulation Settings}
Throughout this section, we consider a spectral band composed of $M=40$ subbands, with each subband having a bandwidth of $B=25$ MHz. Each subband can be occupied by a single user transmitting binary phase shift keying (BPSK) signal, so that the $m$-th signal $s_m(t)$ in (\ref{sig}) is given by
\begin{equation} \label{eq:BPSK}
    s_m(t)=\sqrt{\frac{2E_s}{T_s}}\sum_{n\in\mathbb{Z}} D[n]s\left(t-\frac{n}{B}\right)\cos(2\pi f_mt),
\end{equation}
where $E_s$ represents the symbol energy, chosen uniformly in the range $[1,4]$; $D[n]\in\{-1,1\}$ is BPSK symbol; $s(t)$ is the sinc pulse; $T_s$ denotes symbol duration, $T_s=1/B$; and $f_m$ is the downconverted central frequency of the $m$-th subband. For MCS sampling, $P=20$ cosets are employed, and the sampling frequency is set to $1/LT=25$ MHz so that $L=80>P$. The target FNR level is $\alpha=0.1$. The signal-to-noise ratio (SNR) is defined as the ratio of the signal power to the noise power,  i.e., $\text{SNR}=E_s/N_0$.

The LV extractor is implemented via a convolutional neural network (CNN). To this end, the discrete Fourier transform of the $20\times N_s$ sub-Nyquist sample matrix $\mathbf{Y}$ is separated into its real and imaginary components, obtaining a $40 \times N_s$ real-valued matrix. This matrix is then passed through two one-dimensional (1D) convolutional layers equipped with $1\times 3$ filters, followed by a batch normalization layer, and by a fully-connected layer with sigmoid activation functions. The binary cross-entropy (BCE) is adopted to train the LV extractor in a manner similar to \cite{gao}, \cite{deepsense}. The performance is evaluated via the expected FNR and the TNR, averaged over 50 trials. Error bars illustrate ranges of values covering 95\% of the realized values.

\subsection{Impact of the SNR}

\begin{figure}[h]
    \centerline{\includegraphics[width=11cm]{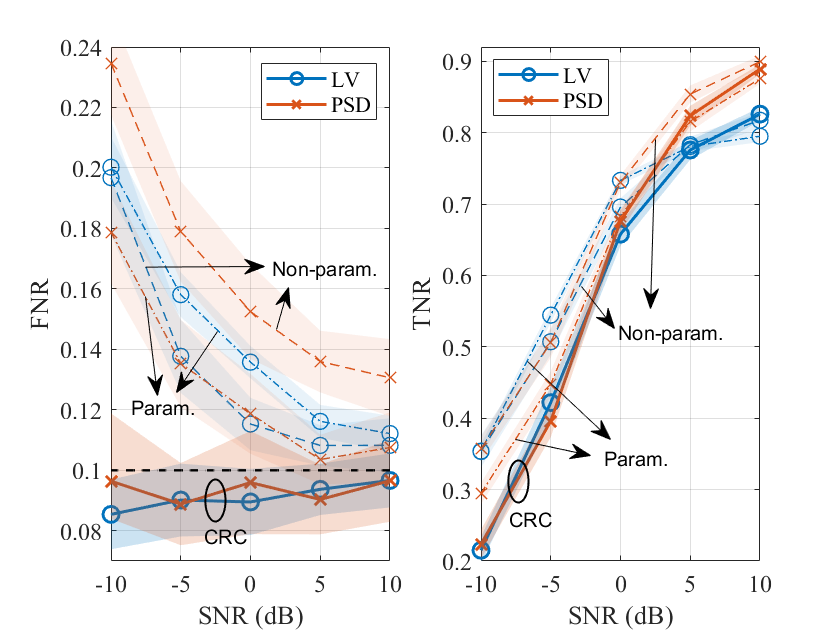}}
    \caption{FNR and TNR as a function of the SNR when the number of calibration data points is $|\mathcal{D}^\text{cal}|=150$, and the number of received samples is \emph{$N= 960$}.}
    \label{calideal}
    
\end{figure}

To start, we investigate the impact of the SNR on FNR and TNR by assuming that the number of occupied subbands is uniformly distributed in the set $\{1,2,\dots,10\}$. The calibration data set consists of $|\mathcal{D}^\text{cal}|=150$ samples, and the total number of received samples is set to $N=960$. Fig. \ref{calideal} plots FNR and TNR as a function of the SNR. The proposed CRC-based scheme ensures for both types of features that the FNR remains below the desired level $\alpha=0.1$ as per Theorem 1. For the existing schemes, the FNR tends to approach the target level as SNR increases, while not strictly meeting the FNR requirement (\ref{goal}) even for SNR levels as large as 10 dB. In this regard, we note that PSD-based feature extraction works better with parameteric thresholding; while LV prefers non-parameteric thresholding.

\subsection{Impact of the Size of the Calibration Set and of the Number of Samples}
\begin{figure}[h]
    \centerline{\includegraphics[width=11cm]{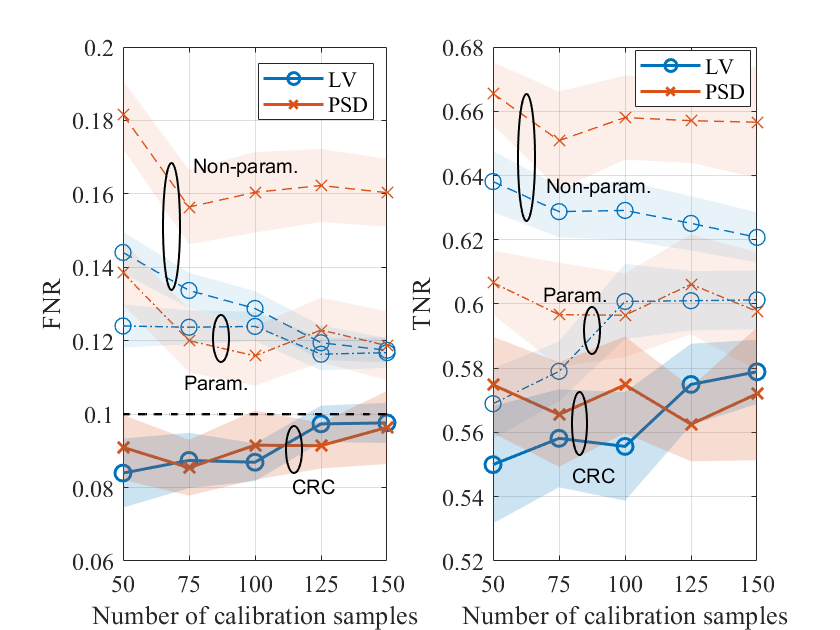}}
    \caption{FNR and TNR as a function of the number of received calibration data point set when the number of received samples is $N=960$, and we set $\text{SNR}=5$ dB.}
    \label{calcal}
\end{figure}

In Fig. \ref{calcal}, the FNR and TNR are depicted as a function of the number of calibration data points under the same condition as Fig. \ref{calideal} except that we set $\text{SNR}=5$ dB, and we consider the more challenging case in which the number of occupied subbands is uniformly selected in the set $\{1,2,\dots,20\}$. The total number of received samples is set to $N=960$. For the existing schemes, even increasing the size of the calibration data set up to 150 does not ensure the FNR requirement (\ref{goal}).

\begin{figure}[h]
    \centerline{\includegraphics[width=11cm]{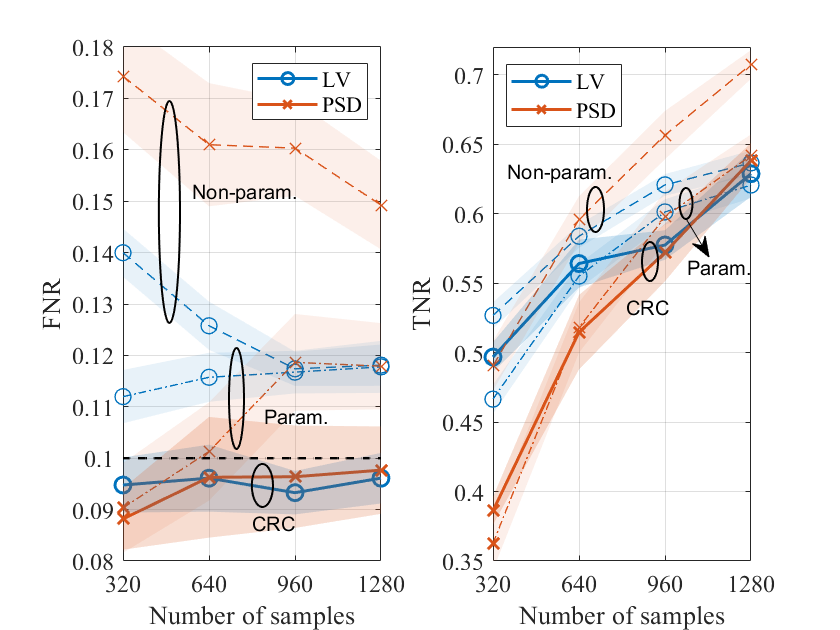}}
    \caption{FNR and TNR as a function of the number of received samples, $N$, when the number of calibration data point set is $|\mathcal{D}^\text{cal}|=150$, and we set $\text{SNR}=5$ dB.}
    \label{calsam}
\end{figure}

In Fig. \ref{calsam}, we vary the number of received samples $N$ for a fixed number $|\mathcal{D}^\text{cal}|=150$ of calibration samples, and we set $\text{SNR}=5$ dB.  As the total number $N$ of received samples increases, the FNR of the non-parameteric schemes becomes closer to the target level with increased TNR. In contrast, the parameteric method suffers from an increased mismatch between the assumed Gaussian distribution and the true distribution of the PSD features. As $N$ grows, given the choice of uniformly distributed transmitted powers, the distribution of each PSD feature tends to be a mixture of Gaussian distributions with decreasing variance as $N$ grows large.

\subsection{Impact of the Sampling Rate}
In Fig. \ref{calcoset}, we investigate the impact of the sampling rate by changing the number of cosets $P$. By increasing $P$, sub-Nyquist sampling becomes increasingly closer to Nyquist sampling, which is recovered for $P=80$. In this experiment, the symbol energy is set to a deterministic value of $E_s=1$ to avoid undesired deficits in parameterized thresholding \cite{vincent}, thus allowing us to focus solely on the impact of the sampling rate. We set $\text{SNR}=0$ dB, and the number of occupied subbands is uniformly distributed in the set $\{6,7,\dots,20\}$.

It is observed that increasing the sampling rate improves spectrum sensing, i.e., lower FNR with higher TNR, for all the schemes, while only the CRC-based thresholding achieves the target FNR irrespective of the sampling rate. The efficacy of the CRC-based thresholding is especially highlighted in the low sampling rate regime ($P < 5$), by noting that all the other schemes fail to achieve the target FNR.

\begin{figure}[h]
    \centerline{\includegraphics[width=11cm]{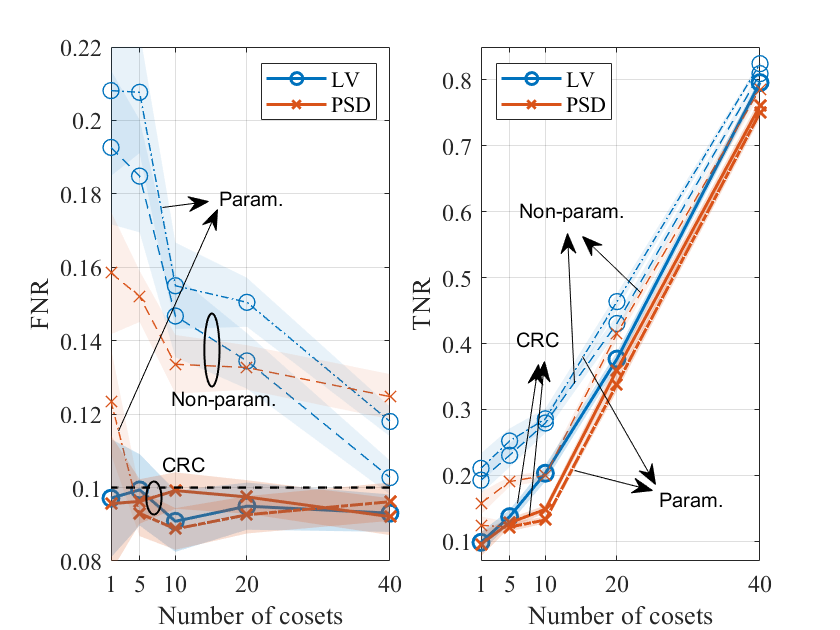}}
    \caption{FNR and TNR as a function of the number of cosets, $P$, when the number of calibration data point set is $|\mathcal{D}^\text{cal}|=150$, and we set $\text{SNR}=0$ dB. The number of samples for each coset is set to $N_s=32$, and the symbol energies are set to a fixed value, $E_s=1$.}
    \label{calcoset}
\end{figure}

\subsection{Impact of Spectral Separation}
Finally, we investigate the impact of the amount of spectral overlap between adjacent subbands replacing the sinc pulse in (\ref{eq:BPSK}) with a raised cosine waveform with the roll-off factor $\beta \in [0,1]$. As $\beta$ increases, the amount of overlap grows larger. The SNR is fixed at 5 dB, and we assume $|\mathcal{D}^\text{cal}|=150$ calibration samples, and $N=960$ sub-Nyquist samples. It is observed in Fig. \ref{calrf} that LV features are more robust to an increased overlap. In particular, using LV features with the proposed CRC-based thresholding scheme not only always guarantees the target FNR level but also maintains a large TNR for all values of the roll-off factor $\beta \in [0,1]$.

\begin{figure}[h]
    \centerline{\includegraphics[width=11cm]{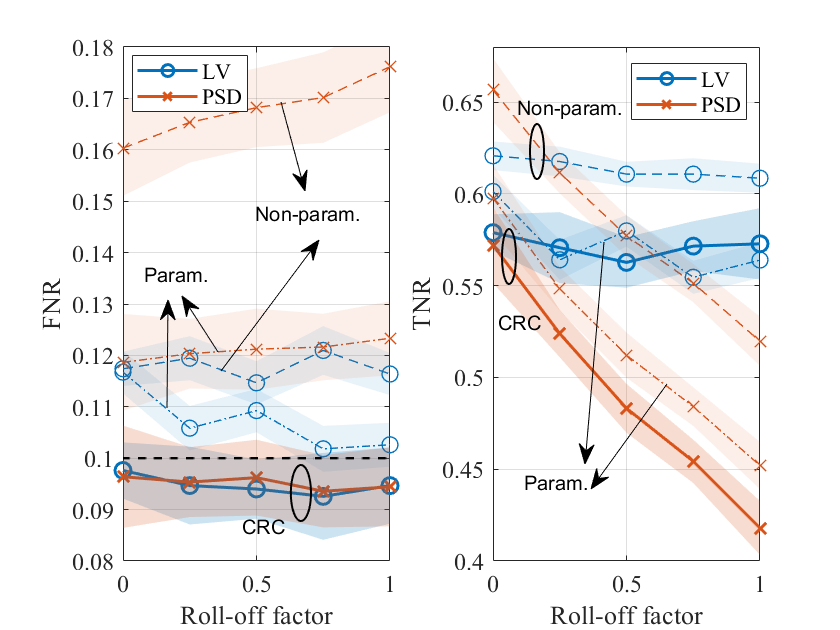}}
    \caption{FNR and TNR as a function of the roll-off factor $\beta$ in the modulating waveform when the number of calibration data point set is $|\mathcal{D}^\text{cal}|=150$, and we set $\text{SNR}=5$ dB.}
    \label{calrf}
\end{figure}

\section{Conclusion}
In this letter, we proposed a thresholding method based on CRC for reliable sub-Nyquist spectrum sensing. The proposed framework can be applied to any sub-Nyquist spectrum sensing technique that extracts per-subband features to provide assumption-free guarantees on the FNR level. Performance benefits are even seen to be particularly pronounced in settings with limited sampling rates. Future work may consider the design of \emph{CRC-aware} sub-Nyquist feature extraction schemes with the aim of further improving the TNR.

\bibliographystyle{ieeetr}
\bibliography{ref}

\end{document}